\documentclass[runningheads]{llncs}
\usepackage[T1]{fontenc}

\usepackage{graphicx}
\graphicspath{{./figures/}}

\usepackage[ruled,vlined,linesnumbered]{algorithm2e}
\usepackage[cmex10]{amsmath}
\usepackage{amssymb}
\usepackage{color}
\newcommand{\cR}{{\cal R}}
\newcommand{\bC}{{\cal C}}
\newcommand{\cH}{{\mathbb{H}}}
\newcommand{\bR}{{\mathbb{R}}}
\newcommand{\cT}{{\cal T}}
\newcommand{\cS}{{\cal S}}
\newcommand{\cC}{{\mathbb{C}}}
\newcommand{\cI}{{\cal I}}
\newcommand{\cV}{{\cal V}}
\newcommand{\mv}{{\sc Mutual Visibility}}

\newcommand{\dist}{{\rm dist}}

\newcommand{\length}{{\rm length}}

\usepackage{amsthm}

\begin{document}
\title{The Mutual Visibility Problem for Fat Robots with Lights}
\author{Rusul J. Alsaedi\orcidID{0000-0002-2942-9519} \and \\
Joachim Gudmundsson\orcidID{0000-0002-6778-7990} \and \\
André van Renssen\orcidID{0000-0002-9294-9947}}
\authorrunning{R. Alsaedi et al.}
\institute{University of Sydney, Australia \\ 
\email{rals2984@uni.sydney.edu.au, joachim.gudmundsson@sydney.edu.au, andre.vanrenssen@sydney.edu.au}}

\maketitle              
\begin{abstract}
Given a set of $n\geq 1$ unit disk robots in the Euclidean plane, we consider the fundamental problem of providing mutual visibility to them: the robots must reposition themselves to reach a configuration where they all see each other. This problem arises under obstructed visibility, where a robot cannot see another robot if there is a third robot on the straight line segment between them. This problem was solved by Sharma {\it et al.}~[ICDCN, 2018] in the luminous robots model, where each robot is equipped with an externally visible light that can assume colors from a fixed set of colors, using 9 colors and $O(n)$ rounds. In this work, we present an algorithm that requires only 2 colors and $O(n)$ rounds. The number of colors is optimal since at least two colors are required even for point robots~[Di Luna {\it et al.}, Information and Computation, 2017].

\end{abstract}
\section{Introduction}
We consider a set of $n$ unit disk robots in $\bR^2$ and aim to position these robots in such a way that each pair of robots can see each other (see Figure~{\ref{fig:16}} for an example initial configuration where not all robots can see each other and an end configuration where they can). This problem is fundamental in that it is typically the first step in solving more complex problems. We consider the problem under the classical oblivious robots model~\cite{Flocchini2012}, where robots are autonomous (no external control), anonymous (no unique identifiers), indistinguishable (no external markers), history-oblivious (no memory of activities done in the past), silent (no means of direct communication), and possibly disoriented (no agreement on their coordinate systems). We consider this problem under the fully synchronous model, where in every synchronized cycle, called a \emph{round}, all robots are activated. All robots execute the same algorithm, following Look-Compute-Move (LCM) cycles~\cite{DAS2016171} (i.e., when a robot becomes active, it uses its vision to get a snapshot of its surroundings (Look), computes a destination point based on the snapshot (Compute), and finally moves towards the computed destination (Move)). We note that the robots do not initially know $n$, the total number of robots in the configuration. 

\begin{figure}[ht]
 \centering
  \includegraphics{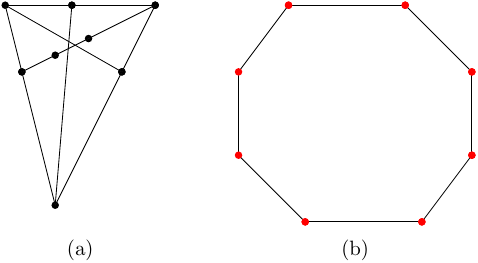}
  \caption{An example of an initial instance (a) and an end configuration (b).}
  \label{fig:16}
\end{figure}

This classical robot model has a long history and has many applications including coverage, exploration, intruder detection, data delivery, and symmetry breaking~\cite{cieliebak2012distributed}. Unfortunately, most of the previous work considered the robots to be dimensionless point robots which do not occupy any space. 

The classical model also makes the important assumption of unobstructed visibility, i.e., any three collinear robots are mutually visible to each other. This assumption, however, does not make sense for the unit disk robots we consider. To remove this assumption, robots under obstructed visibility have been the subject of recent research~\cite{Agathangelou2013,Bolla2012,Chaudhuri15,Cohen2008,Cord-Landwehr2011,Czyzowicz2009,di2017mutual,Luna2014,Luna2014b,Dutta2012,Flocchini2015,Sharma2015b,Sharma2015,Sharma2015c,Vaidyanathan2015}. Under obstructed visibility, robot $r_i$ can see robot $r_j$ if and only if there is at least
one point on the bounding circle of $r_j$ that is visible to $r_i$.

Additionally, a variation on this model received significant attention: the luminous robots model (or robots with lights model)~\cite{di2017mutual,Luna2014,Luna2014b,Peleg2005,Sharma2015b,Sharma2015,Vaidyanathan2015}. In this model, robots are equipped with an externally visible light which can assume colors from a fixed set. The lights are persistent, i.e., the color of the light is not erased at the end of the LCM cycle. When the number of colors in the set is 1, this model corresponds to the classical oblivious robots model~\cite{di2017mutual,Flocchini2012}. In this model, minimizing the number of lights is one of the objectives (in addition to execution time and having few, if any, additional assumptions), as requiring fewer lights would allow for simpler hardware in physical robots. 

Being the first step in a number of other problems, including the Gathering and Circle Formation problems~\cite{sharma2018complete}, the {\mv} problem received significant attention in this new robots with lights model. When robots are dimensionless points, the {\mv} problem was solved in a series of papers~\cite{di2017mutual,Luna2014,Luna2014b,Sharma2015b,Sharma2015,Vaidyanathan2015}. Unfortunately, the techniques developed for point robots do not apply directly to the unit disk robots, due to the lack of collision avoidance. For unit disk robots, much progress has been made in solving the {\mv} problem~\cite{Agathangelou2013,Bolla2012,Chaudhuri15,Cord-Landwehr2011,Czyzowicz2009,Dutta2012,poudel2019sublinear,sharma2018complete,sharma2018make}, however these approaches either require additional assumptions such as chirality (the robots agree on the orientation of the axes, i.e., on the meaning of clockwise), knowledge of $n$, or without avoiding collisions. Additionally, some approaches require a large number of colors and not all approaches bound the number of rounds needed. 

\subsection{Related Work}
Most of the existing work in the robots with lights model considers point robots~\cite{di2017mutual,Luna2014,Sharma2015,Vaidyanathan2015}. 
Di Luna {\it et al.}~\cite{di2017mutual} solved the {\mv} problem for those robots with obstructed visibility in the lights model, using 2 and 3 colors under semi-synchronous and asynchronous computation, respectively. 
Sharma {\it et al.}~\cite{Sharma2015} provided a solution for point robots that requires only 2 colors, which is optimal since at least two colors are needed~\cite{di2017mutual}. Unfortunately, the required number of rounds is not analyzed. 
Sharma {\it et al.}~\cite{sharma2017constant} also considered point robots in the robots with lights model. In the asynchronous setting, they provide an $O(1)$ time and $O(1)$ colors solution using their Beacon-Directed Curve Positioning technique to move the robots. 

Mutual visibility has also been studied for fat robots. Agathangelou {\it et al.}~\cite{Agathangelou2013} studied it in the fat robots model of Czyzowicz {\it et al.}~\cite{Czyzowicz2009}, where robots are not equipped with lights. Their approach allows for collisions, assumes chirality, and the robots need to know $n$, making it unsuited for our setting. 
Sharma {\it et al.}~\cite{sharma2018make} developed an algorithm that solves coordination problems for fat robots in $O(n)$ rounds in the classical oblivious model, assuming $n$ is known to the robots.

Poudel {\it et al.}~\cite{poudel2019sublinear} studied the {\mv} problem for fat robots on an infinite grid graph $G$ and the robots have to reposition themselves on the vertices of the graph $G$. They provided two algorithms; the first one solves the {\mv} problem in $O(\sqrt{n})$ time under a centralized scheduler. The second one solves the same problem in $\Theta(\sqrt{n})$ time under a distributed scheduler, but only for some special instances. 

When considering both fat robots and the robots with lights model, the main result is by Sharma {\it et al.}~\cite{sharma2018complete}. Their solution uses 9 colors and solves the {\mv} problem in $O(n)$ rounds. 

\subsection{Contributions}
We consider $n \geq 1$ unit disk robots in the plane and study the problem of providing mutual visibility to them. We address this problem in the lights model. In particular, we present an algorithm that solves the problem in $O(n)$ rounds using only 2 colors while avoiding collisions. The number of colors is optimal since at least two colors are needed for point robots~\cite{di2017mutual}. 

Our algorithm works under fully synchronous computation, where all robots are activated in each round and they perform their LCM cycles simultaneously in synchronized rounds. The moves of the robots are rigid, i.e., they cannot be interrupted during the execution, for example by an adversary~\cite{Flocchini2012}.

Our results improve on previous work in two ways. First, we improve in terms of the number of colors used compared to~\cite{sharma2018complete}. Secondly, by using fat robots and having a linear number of rounds, we generalize the results known for point robots~\cite{di2017mutual,Sharma2015}. Additionally, we require no additional assumptions such as chirality or knowledge of $n$.

\section{Preliminaries}
Consider a set of $n\geq 1$ anonymous robots $\cR=\{r_1,r_2,\ldots,r_n\}$ operating in the Euclidean plane. During the entire execution of the algorithm, we assume that $n$ is \emph{not} known to the robots.
Each robot $r_i\in \cR$ is a non-transparent disk with diameter $1$, sometimes referred to as a fat robot. The center of the robot $r_i$ is denoted by $c_i$ and the position of $c_i$ is also said to be the position of $r_i$. We denote by $\dist(r_i,r_j)$ the Euclidean distance between the two robots, i.e., the distance from $c_i$ to $c_j$. To avoid collisions among robots, we have to ensure that $\dist(r_i,r_j)\geq 1$ between any two robots $r_i$ and $r_j$ ($i \neq j$) at all times.
Each robot $r_i$ has its own coordinate system, and it knows its position with respect to its coordinate system.
Robots may not agree on the orientation of their coordinate systems, i.e., there is no common notion of direction. Since all the robots are of unit size, they agree implicitly on the unit of measure of other robots. The robots have a camera to take a snapshot, and the visibility of the camera is unlimited provided that there are no obstacles (i.e., other robots)~\cite{Agathangelou2013}.

We say that a point $p$ in the plane is visible by a robot $r_i$ if there is a point $p_i$ in the bounding circle of $r_i$ such that the
straight line segment $\overline{p_ip}$ does not intersect any other robot.
Following the fat robot model~\cite{Agathangelou2013,Czyzowicz2009},
we assume that a robot $r_i$ can see another robot $r_j$ if there is at least
one point on the bounding circle of $r_j$ that is visible from $r_i$.
We say that robot $r_i$ fulfills the mutual visibility property if $r_i$ can see all other robots in $\cR$. Two robots $r_i$ and $r_j$ are said to {\em collide} at time $t$ if the bounding circles of $r_i$ and $r_j$ share a common point at time $t$. For simplicity, we use $r_i$ to denote both the robot $r_i$ and the position of its center $c_i$.

Each robot $r_i$ is equipped with an externally visible light that can assume any color from a fixed set $\bC$ of colors. The set $\bC$ is the same for all robots in~$\cR$. 
The color of the light of robot $r$ at time $t$ can be seen by all robots that are visible to $r$ at time $t$.

A {\it configuration} $\cC$ is a set of $n$ tuples in $\bC\times \bR^2$ which define the colors and positions of the robots. Let $\cC_t$ denote the configuration at time $t$. Let $\cC_t(r_i)$ denote the configuration $\cC_t$ for robot $r_i$, i.e., the set of tuples in $\bC\times \bR^2$ of the robots visible to $r_i$.
A configuration $\cC_t$ is {\em obstruction-free} if for all $r_i \in \cR$, we have that $|\cC_t(r_i)|=n$. In other words, when all robots can see each other. 

Let $\cH_t$ denote the convex hull formed by the robots in $\cC_t$.
Let $\partial\cH_t=\cV_t\cup \cS_t$ denote the set of robots on the boundary of $\cH_t$, where $\cV_t \subseteq \cR$ is the set of corner robots lying on the corners of $\cH_t$ and $\cS_t \subseteq \cR$ is the set of robots lying in the interior of the edges of $\cH_t$. The robots in the set $\cV_t$ are called {\em corner robots} and those in the set $\cS_t$ are called {\em side robots}.
The robots in the set $\cI_t = \cH_t\backslash \partial\cH_t$
are called {\em interior robots}. Given a robot $r_i\in \cR$, we denote by $\cH_t(r_i)$ the convex hull of $\cC_t(r_i)$.
Note that $\cH_t(r_i)$ can differ from $\cH_t$ if $r_i$ does not see all robots on the convex hull. 

Given two points $a,b\in \bR^2$, 
we denote by $|{\overline{ab}}|$ the length of the straight line segment $\overline{ab}$ connecting them. Given $a,b,d\in \bR^2$, we use $\angle abd$ to denote the counterclockwise angle at point $b$ between $ab$ and $bd$.

At any time $t$, a robot $r_i\in \cR$ is either active or inactive. When active, $r_i$ performs a
sequence of {\em Look-Compute-Move} (LCM) operations:
\begin{itemize}
\item {\em Look:} a robot takes a snapshot of the positions of the robots visible to it in its own coordinate system; 
\item {\em Compute:} executes its algorithm 
using the snapshot. This returns a destination point $x\in \bR^2$ and a color $c\in \bC$; and
\item {\em Move:}  moves to the computed destination $x\in \bR^2$ (if $x$ is different than its current position) and sets its own light to color $c$.
\end{itemize}

We assume that the execution starts at time $0$. Therefore, at time $t=0$, the robots start in an arbitrary configuration $\cC_0$ with $\dist(r_i,r_j)\geq 1$ for any two robots $r_i,r_j\in \bR^2$, and the color of the light of each robot is set to {\it Off}.

Formally, the {\mv} problem is defined as follows: Given any $\cC_0$, 
in a finite number of rounds, reach an obstruction-free configuration without having any collisions in the process.
An algorithm is said to solve the {\mv} problem if it always achieves an obstruction-free configuration from any arbitrary initial configuration in a finite number of rounds. Each robot executes the same algorithm locally every time it is activated. We measure the quality of the algorithm both in terms of the number of colors and the number of rounds needed to solve the {\mv} problem.

Finally, we need the following definitions to present our {\mv} algorithm.
Let $e=\overline{v_1v_2}$ be a line segment connecting two corner robots $v_1$ and $v_2$ of $\cH_t$. Following Di Luna {\it et al.}~\cite{Luna2014}, we define the \emph{safe zone} $S(e)$ as a non-empty portion of the plane outside $\cH_t$ such that the corner robots $v_1$ and $v_2$ of $\cH_t$ remain corner robots when a side robot is moved into this area: for all points $x\in S(e)$, we ensure that $\angle xv_1v_2 \leq \frac{180^\circ-\angle v_0v_1v_2}{4}$ and $\angle v_1v_2x \leq \frac{180^\circ-\angle v_1v_2v_3}{4}$, where $v_0, v_1, v_2,$ and $v_3$ are consecutive vertices of the convex hull of $\cH_t$ (see Figure~\ref{fig:safe zone}(a))\footnote{The division by 4 ensures that no robots can become collinear. Values other than 4 can also work.}. 

We note that side robots and interior robots may not always be able to compute $S(e)$ exactly due to obstructions of visibility leading to different local views. A single side robot on $e$ can compute $S(e)$ exactly. However, when there is more than one robot on $e$, $S'(e)$ is the safe region computed by a robot based on its local view. It is guaranteed that $S'(e)\subseteq S(e)$ (see Figure~\ref{fig:safe zone}(b) for the safe zone of robot $r_2$, which cannot see $v_1$ and thus uses $r_1$ and $r_3$ to compute a more restricted safe zone).

\begin{figure}[ht]
 \centering
  \includegraphics{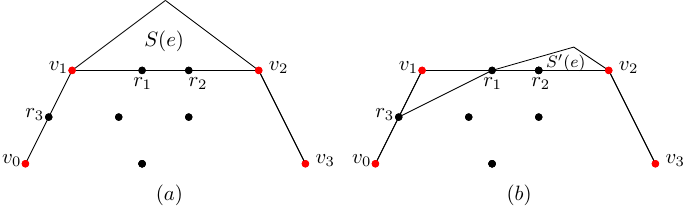}
  \caption{(a) The safe zone of $e=\overline{v_1v_2}$. (b) The safe zone of a side robot $r_2$ on $e$.}
  \label{fig:safe zone}
\end{figure}

Unfortunately, interior robots force us to use a slightly modified definition of a safe zone compared to Di Luna {\it et al.}~\cite{Luna2014}. As our algorithm will later show, we only use the safe zone of an edge $e$ for the interior robot $r_1$ that is closest to that edge. This implies that $r_1$ can always see both endpoints of $e$. However, $r_1$ may not be able to see $v_0$ and/or $v_3$ due to other interior robots blocking visibility to them. Moreover, if $r_1$ observes an interior robot $r_2$ between two corner robots in cyclic order (say immediately counterclockwise from $v_1$), it has no way of checking whether there exists a corner robot that is hidden from $r_1$'s view by $r_2$. To overcome this issue, we will (pessimistically) assume that $r_2$ indeed blocks visibility to a corner robot and to minimize the implied safe zone defined using this hidden corner robot, we will assume this robot is infinitely far away from $r_1$ in the direction of $r_2$. This means that the line segment connecting this potential corner robot to $v_1$ is parallel to $\overline{r_1 r_2}$. Hence, we use the line parallel to $\overline{r_1 r_2}$ through $v_0$ to determine the angle allowed for the safe zone, i.e., $\angle v_0v_1v_2$ is the angle between edge $e$ and the line parallel to $\overline{r_1 r_2}$ through $v_0$ (see Figure~\ref{fig:safe zone2}). 

\begin{figure}[ht]
 \centering
  \includegraphics{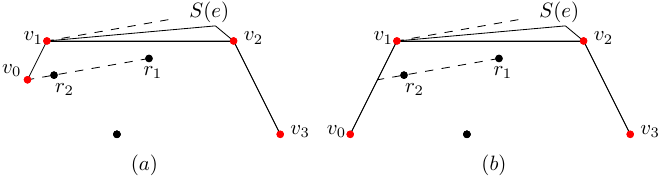}
  \caption{Robot $r_1$ cannot determine whether robot $r_2$ blocks visibility to a corner robot. In either case the line parallel to $\overline{r_1 r_2}$ is used to compute the safe zone. (a) Robot $r_2$ hides a corner robot. (b) Robot $r_2$ does not hide a corner robot.}
  \label{fig:safe zone2}
\end{figure}

\section{The Mutual Visibility Algorithm}
In this section, we present an algorithm that solves the {\mv} problem for $n\geq 1$ unit disk robots under rigid movement in the robots with lights model. Our algorithm assumes the fully synchronous setting of robots. The algorithm needs two colors: $\bC=\{${\it Off}, {\it Red}$\}$. A red robot represents a corner robot. A robot whose light is off represents any other robot. 
See Figure~\ref{fig:17} for an example. 
Initially, the lights of all robots are off. 

\begin{figure}[ht]
 \centering
  \includegraphics{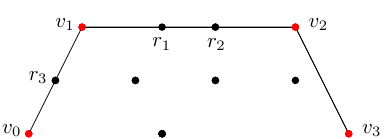}
  \caption{The different colors of the robots: corner robots (red), side robots (off), and interior robots (off).}
  \label{fig:17}
\end{figure}

It has been shown that positioning the robots
in the corners (i.e., vertices) of an $n$-vertex convex polygon provides a solution to the {\mv} problem~\cite{di2017mutual,Luna2014,Flocchini2015,Sharma2015b,Sharma2015,Vaidyanathan2015}. Hence, our algorithm also ensures that the robots eventually position themselves in this way. 

Conceptually, our general strategy consists of two phases, though the robots themselves do not explicitly discern between them. In the \emph{Side Depletion} phase, some side robots move to become corner robots, ensuring that there are only corner and interior robots left. In the \emph{Interior Depletion} phase, interior robots move and become corner robots. The move-algorithm checks if the robot's path shares any point with any other robots, ensuring that no collision occur. Throughout both phases, corner robots slowly move to expand the convex hull to ensure that the interior robots can move through the edges of the convex hull when needed. This movement is deterministic and is taken into account when moving robots to become corners of the expanding hull. 

Detailed pseudocode of the algorithm and its subroutines can be found in the appendix. 

\subsection{The Side Depletion Phase}
The first phase of our algorithm is the Side Depletion (SD) phase. During this phase, every robot first determines if it is a corner, side, or interior robot and sets its light accordingly. Note that robots can make this distinction themselves, by checking what angle between consecutive robots it sees: if some angle is larger than $180^\circ$ it is a corner robot, if the angle is exactly $180^\circ$ it is a side robot, and otherwise it is an interior robot. 

In every round, all corner robots move a distance of 1 along the angle bisector determined by its neighbors in the direction that does not intersect the interior of the convex hull. In other words, in each round, the corner robots move to expand the size of the convex hull. We note that since all corner robots move this way, they all stay corner robots throughout this process. 

Side robots that see at least one corner robot (i.e., a robot with a red light) move to become new corner robots of $\cH$ (using the safe zone described earlier and taking the above movement of corner robots into account) and change their light to red. Side robots that do not see a corner robot on their convex hull edge do not move and will become interior robots in the next round (due to the change to the convex hull), while keeping their light off. 

More precisely, a side robot $r$ on edge $e=\overline{v_1v_2}$ of $\cH_k$ moves as follows:
If at least one of its neighbors on $\overline{v_1v_2}$ is a corner robot, $r$ moves to a point in the safe zone $S(e)$. There are at most two such robots $r_1$ and $r_2$ on each edge $\overline{v_1v_2}$ (see Figure~{\ref{fig:1} and \ref{fig:2}}). Sharma {\it et al.}~\cite{sharma2018complete} showed that these can move simultaneously to the safe zone outside the hull. Both $r_1$ and $r_2$ become new corners of $\cH$ and change their lights to red (see Figure~{\ref{fig:2}}).

\begin{figure}[ht]
 \centering
  \includegraphics{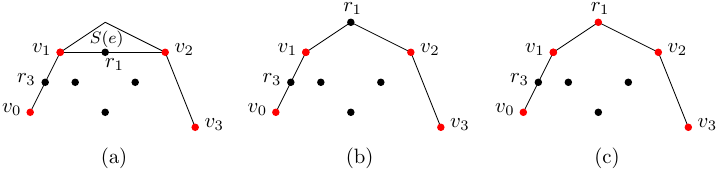}
  \caption{One side robot $r_1$ on an edge $e=\overline{v_1v_2}$ moves to become a corner robot.}
  \label{fig:1}
\end{figure}

\begin{figure}[ht]
 \centering
  \includegraphics{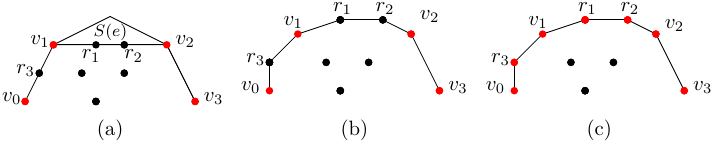}
  \caption{Two side robots $r_1$ and $r_2$ on an edge $e=\overline{v_1v_2}$ move to become corner robots.}
  \label{fig:2}
\end{figure}

If both of its neighbors on $\overline{v_1v_2}$ are not corners (see Figure {\ref{fig:6}}), robot $r$ does not move and stay in its place, and it will become an interior robot in the next round.

We only execute this phase once, at the start of our algorithm and only move each robot once.

\begin{figure}[ht]
  \centering
   \includegraphics{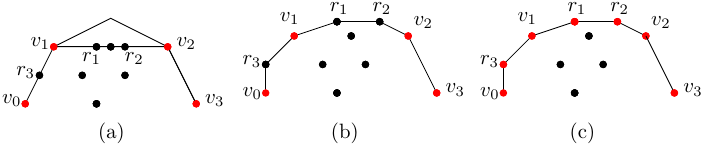}
   \caption{When there are more than two side robots on an edge of the convex hull, only two side robots on the edge move to become corner robots. These are the clockwise and the counterclockwise extreme side robots. In this case, robots $r_1$ and $r_2$ move to become corner robots.}
  \label{fig:6}
\end{figure}

\subsection{The Interior Depletion Phase}
Once the SD phase finishes, the Interior Depletion (ID) phase starts. During this phase the robots in the interior of the hull move such that they become new vertices of the hull. 

In every round, all corner robots move as in the SD phase, expanding the convex hull. This ensures that the length of all edges increases and thus interior robots can move through these edges to in turn become corner robots themselves. All movement described in the remainder of this paper takes the (predictably) expanding convex hull into account. 

Next we describe how an interior robot moves. Given a robot $r_i$, we define its \emph{eligible} edges as those edges of length at least 3 for which no other robot is closer to the edge\footnote{The length of 3 is used to ensure that two robots can move through the same edge without colliding with each other (requiring a length of 2) while ensuring that they also do not collide with the corner robots on the edge (adding a length of 0.5 per corner robot).} and $r_i$ is not between two other robots at the same distance to this edge. 
The interior robots start by determining their eligible edges (see Figure~{\ref{fig:10}}(a)). In the figure, robot $r_i$ finds edges $\overline{v_1v_2}$ and $\overline{v_2v_3}$ eligible, whereas $r_j$ finds $\overline{v_2v_3}$, and $r_l$ finds $\overline{v_3v_4}$ eligible. However, the robots between $r_i$, $r_j$ find no edge eligible. Let $Q$ denote the set of edges that are eligible to an interior robot $r_i$.
Every interior robot that has an eligible edge moves perpendicular to one of its eligible edges $e$ towards $e$ to become a corner robot by moving through $e$ (see Figure~{\ref{fig:10}}(b)), while avoiding collisions with other robots (see Figure~{\ref{fig:10}}(c)). If the path is clear, it moves outside the hull into its safe zone to become a new corner as described earlier (see Figure~{\ref{fig:10}}(d)) and changes its color to red (see Figure~{\ref{fig:10}}(e)).

\begin{figure}[ht]
 \centering
  \includegraphics{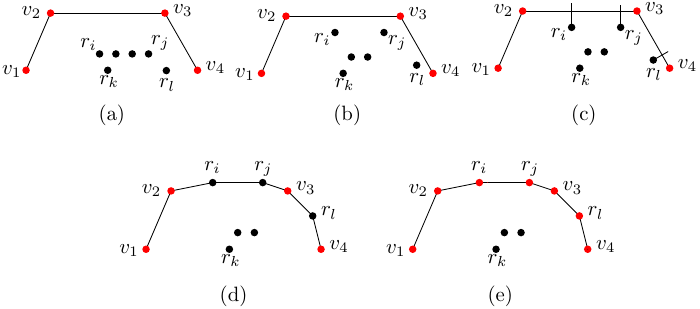}
  \caption{(a) The eligible edge computation. The robot $r_i$ finds edges $\overline{v_1v_2}$ and $\overline{v_2v_3}$ eligible, whereas $r_j$ finds $\overline{v_2v_3}$, and $r_l$ finds $\overline{v_3v_4}$ eligible. The robots between $r_i$, $r_j$ find no eligible edges.
  (b) Interior robots $r_i$, $r_j$ and $r_l$ move towards the edge.
  (c) Since interior robots $r_i$, $r_j$ and $r_l$ move perpendicular to their respective edge, collisions with other robots are avoided.
  (d) After the interior robots $r_i$, $r_j$ and $r_l$ move, they become corners.
  (e) Robots $r_i$, $r_j$ and $r_l$ change their lights to red.}
  \label{fig:10}
\end{figure}

When both phases are finished, the {\mv} problem is solved, and all the robots are in the corners of the convex hull with red lights. 

\subsection{Special Cases}
There are two special cases to consider: $n=1$, and the case where the initial configuration is a line. The case $n = 1$ can be easily recognized by the only robot, since it does not see any other robot and thus it can terminate. 

If in the initial configuration all robots lie on a single line, we differentiate between the robots that see only one other robot and the robots that see two other robots. If a robot $r_i$ sees only one other robot $r_j$, when $r_i$ is activated for the very first time it sets its light to red and moves orthogonal to the line $\overline{r_ir_j}$ for some arbitrary positive distance. When $r_i$ is activated in future rounds and $\cH_k(r_i)$ is still a line segment, it can conclude that there are only two robots and it does nothing until it sees $r_j$ set its light to red. Once $r_j$ sets its light to red, $r_i$ terminates. 

If a robot $r_i$ sees two other robots $r_j$ and $r_l$, robot $r_i$ will be able to tell if $\cH_k(r_i)$ is a line segment as follows. Robot $r_i$ will move orthogonal to line $\overline{r_j r_l}$ and set its light to red if and only if it sees that the lights of $r_j$ and $r_l$ are set to red, as this indicates that both other robots see only a single other robot, i.e., $n=3$. Otherwise, moving the two extremal robots of the initial configuration as described above ensures that the configuration is no longer a line segment, allowing the SD and ID phase to solve the problem. 

As these special cases add only a constant number of rounds to the running time and do not influence the number of colors, we focus on the general case in the remainder of this paper.

\section{Analysis}
We proceed to prove that our algorithm solves the {\mv} problem in a linear number of rounds, using only two colors and while avoiding collisions between the robots. 

We start with some properties of the Side Depletion phase.

\begin{lemma}
\label{lemma:1}
Given a configuration $\cC_{k}$ and an edge $e=\overline{v_1v_2}$ of $\cH_{k}$, if a robot $r_i\in e$ moves away from $e$, it will move into the safe zone $S(e)$. 
\end{lemma}
\begin{proof}
We prove this lemma using proof techniques similar to those of Lemma 3 in~\cite{Luna2014}. Let $v_1$ and $v_2$ be the two corner robots that define $e$. If there is a single robot $r\in e$, $r$ can compute $S(e)$ exactly and then move into $S(e)$, proving the lemma. Consider the situation when there are at least two side robots on $e$. Let $r_1$ and $r_2$ be the two robots on $e$ that are neighbors of $v_1$ and $v_2$, respectively. In the fully synchronous setting, both $r_1$ and $r_2$ move from $e$ in the same round. Consider only the move of $r_2$ to $S(e)$ (the move of $r_1$ follows similarly). 

Robot $r_2$ orders the robots it can see in clockwise order and let this ordering be $\{v_0,v,r_2,v_2,v_3\}$, where $v_0$ is the first robot non-collinear to $r_2$ in the clockwise direction with its light set to red, $v$ is the robot that is collinear with $r_2$ in the clockwise direction, and $v_2$ is the collinear robot in the counterclockwise direction with its light set to red, and $v_3$ is the first non-collinear robot in the counterclockwise direction with its light set to red. Following the rules of Algorithm~\ref{algorithm:2s}, $r_2$ computes $\alpha=180^\circ-\angle {v_0 v v_2}$, $\beta=180^\circ-\angle {v v_2 v_3}$, and $\delta=\min\{\alpha/4, \beta/4\}$. We note that since we calculate $\alpha$ by subtracting $\angle {v_0 v v_2}$ from $180^\circ$, $\alpha$ may be smaller than the actual angle used to define $S(e)$. Therefore, any point $x$ in the safe zone computed by $r_2$ is inside the safe zone of $e$, and thus, $r_2$ will move inside $S(e)$. The same holds for $r_1$. The other robots on $e$ between $r_1$ and $r_2$ do not move.
\end{proof}

\begin{lemma}
\label{lemma:2}
Let $r_i$ and $r_j$ be the robots that are neighbors of endpoints $v_1$ and $v_2$ on edge $e$, respectively. When there are $p \leq 2$ side robots on $e$, $r_i$ and $r_j$ become corners and change their light to red in the next round. When there are $p > 2$ side robots on $e$, $r_i$ and $r_j$ become corners and change their light to red after which all the robots on $e$ between $r_i$ and $r_j$ lie inside the convex hull and become interior robots.
\end{lemma}
\begin{proof}
If $p \leq 2$, $r_i$ and $r_j$ see only the corners and each other on $e$. Hence, both robots move and by Lemma~\ref{lemma:1} they move into $S(e)$. By moving $r_i$ and $r_j$ to $S(e)$, they become corner robots, as was also argued by Di Luna {\em et al.}~\cite{Luna2014}. 

When $p>2$, a similar argument shows that both $r_i$ and $r_j$ become corners of $\cH_k$ after they move once and change their light to red in the next round. The other side robots on $e$ remain in their places and since $\cC_{k}$ is not a line, moving $r_i$ and $r_j$ creates a hull that has more than three sides, implying that the robots between $r_i$ and $r_j$ lie strictly inside this hull. Thus, the other side robots become interior robots.
\end{proof}

\begin{lemma}
\label{lemma:3}
Given a configuration $\cC_{0}$ with $q \geq 1$ side robots. After one round, all side robots become either corner robots or interior robots.
\end{lemma}
\begin{proof}
The movements of side robots on different edges of $\cH$ do not interfere with each other. Therefore, we prove this lemma for a single edge $e$ and the same argument applies for the side robots on other edges of $\cH$.

When there is only one robot $r$ on $e$, then $r$ can compute $S(e)$ exactly and move to a point $x \in S(e)$ as soon as it is activated. When there are two or more robots on $e$, two side robots (the extreme ones on this edge) become corners in one round by Lemma~\ref{lemma:2}. This causes the other robots on $e$ to become interior robots. 

Since the robots on different edges do not influence each other and the moves on any edge end in one round, this phase ends in one round.
\end{proof}

Now that there are no more side robots, we argue that the interior robots also eventually become corners.
We first show that the interior robots can determine whether the SD phase has finished. 
\begin{lemma}
\label{lemma:4}
Given a configuration $\cC_{k}$ and an edge $e=\overline{v_1v_2}$ of $\cH_k$, no robot in the interior of $\cH_k$ moves to $S(e)$ if there is a side robot on $e$.
\end{lemma}
\begin{proof}
If there are side robots in $\cH_k$, it is easy to see that every corner robot of $\cH_k$ on an edge that contains side robots sees at least one side robot. Similarly, when there are side robots, interior robots can easily infer that the SD phase is not finished, and hence they do not move to their respective $S(e)$.
\end{proof}

Next, we argue in a series of lemmas that every interior robot will eventually become a corner robot and it does not collide with any robots in doing so. Let $\cC_{SD}$ denote the configuration of robots after the SD phase is finished and let $\cH_{SD}$ be the convex hull created by $\cC_{SD}$.

\begin{lemma}
\label{lemma:6}
Let $I_k$ be the set of interior robots in round $k \in \mathbb{N}^+$. In each round $k$ until $I_k=\emptyset$, if there is an edge of length at least 3, there is at least one robot in $I_k$ for which the set of line segments $Q$ is not empty.
\end{lemma}
\begin{proof}
We note that every edge of the convex hull of the corner robots $\cH_k$ is closest to some interior robot(s). In particular, this holds for any edge of length at least 3. We note that this set of interior robots forms a line, as they all have the same closest distance to the edge. Out of these robots, by definition, the left and right extreme ones have the edge in their $Q$. 
\end{proof}

\begin{lemma}
\label{lemma:8}
Let $\cC_{SD}$ be the configuration after the SD phase ended and let $e=\overline{v_1v_2}$ be the edge of $\cH_{SD}$ closest to some interior robot $r_i$. If the robot $r_i \in I_k$ moves, it moves inside the safe zone $S(e)$.
\end{lemma}
\begin{proof}
We prove this lemma using the proof technique similar to the proof of Lemma 3 in~\cite{Luna2014}. When there is a single closest interior robot $r \in I_k$, $r$ can compute the region $S(e)$ and move to it, proving the lemma. Consider now the situation when there are at least two closest interior robots. Let $r_1$ and $r_2$ be two of these robots. Since we work in the fully synchronous setting, both $r_1$ and $r_2$ move at the same time. Consider only the move of $r_2$ (the move of $r_1$ follows similarly). Robot $r_2$ orders the corner robots that are visible to it according to its local notion of clockwise direction and let this ordering be $\{v_0,v_1,v_2,v_3\}$, where $v_0$ is a corner robot preceding $v_1$ in the clockwise direction and $v_3$ is the corner robot following $v_2$ in clockwise direction. Following the rules of our algorithm, $r_2$ computes $\alpha=180^\circ-\angle {v_0v_1v_2}$, $\beta=180^\circ-\angle v_1v_2v_3$, and $\delta=\min\{\alpha/4, \beta/4\}$ (see Figure~{\ref{fig:safe zone}}). We note that since we calculate $\alpha$ by subtracting $\angle {v_0v_1v_2}$ from $180^\circ$, $\alpha$ is in fact a lower bound on the actual angle that any robot in $I_k$ at the same distance from edge $e$ will compute. Let $x'$ be the nearest to $e$ in the safe zone outside the convex hull such that either $\angle {x'v_1v_2} = \delta$ or $\angle x'v_2v_3 = \delta$ and define $x=x'+\overline{r_2m}$, where $m$ is the intersection point of $e$. The same holds for $r_1$. Our algorithm guarantees that in every round at most two closest interior robots to an edge can move through this edge. 
\end{proof}

\begin{lemma}
\label{lemma:9}
Given any initial configuration $\cC_0$, no collisions of robots occur until $I_k=\emptyset$.
\end{lemma}
\begin{proof}
This lemma is proved by considering Algorithm~\ref{algorithm:2i}. An interior robot $r_i$ with light off does not collide with any other interior robot since the move of $r_i$ is perpendicular to the closest edge $\overline{r_1r_2}$ and there is sufficient space on the edge for the robot to move through it. The robots moving through different edges of $\cH_k$ do not collide since those robots are the closest robots to those edges because the $S(e)$ of different edges are disjoint.
\end{proof}

\begin{lemma}
\label{lemma:10}
There exists an integer $k \in \mathbb{N}^+$ such that the robots in $I_k$ closest to their eligible edge are able to move outside the convex hull $\cH_k$ and become corner robots with their light set to red.
\end{lemma}
\begin{proof}
By Lemma~\ref{lemma:9}, the robot $r_i$ does not collide with other interior robots while it tries to move toward the edge $\overline{v_1v_2}$ of $\cH_k$. Since there is no side robot after the first round by Lemma~\ref{lemma:3}, those cannot block $r_i$'s movement.
By Lemma~\ref{lemma:9}, there is no collision for robot $r_i$ while it passes $e = \overline{v_1v_2}$ where $v_1$ and $v_2$ are the endpoints of the edge that $r_i$ passes through to its computed point in $S(e)$. Since the movements are rigid, $r_i$ reaches its computed point in the safe zone once it moves and changes its color to red.
\end{proof}

\begin{lemma}
\label{lemma:11}
Given any initial configuration $\cC_0$, there exists an integer $k \in \mathbb{N}^+$ such that $I_k=\emptyset$ in $\cC_k$ and the corner robots do not move in any round $k'>k$.
\end{lemma}
\begin{proof}
When $I_k \neq \emptyset$ each corner robot sees at least one robot with light off. Therefore, combining the results of Lemmas~\ref{lemma:6},~\ref{lemma:8},~\ref{lemma:9}, and~\ref{lemma:10} with this observation, we have that, given any $\cC_0$, there is some round $k \in \mathbb{N}^+$ such that $I_k=\emptyset$. 

Corner robots do not move after $I_k=\emptyset$, since they do not see robots with light off, thus terminating.
\end{proof}

\begin{theorem}
\label{lemma:12}
Given any initial configuration $\cC_0$, there is some round $k \in \mathbb{N}^+$ such that all robots lie on $\cH_k$ and have their lights set to red.
\end{theorem}
\begin{proof}
Lemma~\ref{lemma:11} shows that there exists a round $k$ such that there are no interior robots left. 
Interior robots that moved to become corner robots changed their lights to red as soon as they reached their corner positions. Furthermore, the interior robots move to the safe zone where they by definition become corners. Since Lemma~\ref{lemma:11} guarantees that there are no collisions, the robots occupy different positions of $\cH_k$ and all their lights will be red.
\end{proof}

Next, we argue that the robots can determine when there are no interior robots left. 

\begin{lemma}
\label{lemma:14}
If there exists a robot with light off, there is at least one interior robot that is visible to any corner robot $r_i$.
\end{lemma}
\begin{proof}
If there is at least one interior robot, every corner robot can see some interior robot (for example the one closest to it). By definition, every interior robot has its light off, proving the lemma. 
\end{proof}

\begin{lemma}
\label{lemma:15}
Given a robot $r_i \in \cR$ with its light set to red and a round $k \in \mathbb{N}^+$, if all robots in $\cC_k(r_i)$ have their light set to red, and no robot is in the interior of $\cH_k(r_i)$, then $\cC_k$ does not contain interior robots.
\end{lemma}
\begin{proof}
When all the robots in $\cC_k(r_i)$ have their light set to red, this means that there is no robot with light off. Since any interior robot would have color off and by Lemma~\ref{lemma:14} at least one of these robots would be visible to $r_i$, this proves the lemma. 
\end{proof}

We are now ready to prove that the {\mv} problem is solvable using only two colors. 
Let $\cC_{ID}$ denote the configuration of robots after the ID phase is finished and let $\cH_{ID}$ be the convex hull created by $\cC_{ID}$.

\begin{theorem}
\label{lemma:16}
The {\mv} problem is solvable without collisions for unit disk robots in the fully synchronous setting using two colors in the robots with lights model.
\end{theorem}
\begin{proof}
We have from Lemma~\ref{lemma:3} that from any initial non-collinear $\cC_0$, we reach a configuration $\cC_{SD}$ without side robots after one round, some becoming corner robots and some becoming interior robots. 
Once the SD phase is over, Theorem~\ref{lemma:12} shows that the ID phase moves all interior robots to become corner robots. 
We have from Lemma~\ref{lemma:15} that robots can locally detect whether the ID phase is over and configuration $\cC_{ID}$ is reached.
By Lemma~\ref{lemma:9}, no collisions occur in the SD and ID phases.

Therefore, starting from any non-collinear configuration $\cC_0$, all robots eventually become corners of the convex hull, solving the {\mv} problem without collisions. 

It remains to show that starting from any initial collinear configuration $\cC_0$ the robots correctly evolve into some non-collinear configuration from which we can apply the above analysis. If $n \leq 3$, this can be shown through a simple case analysis: For $n=1$, when the only robot becomes active, it sees no other robot, changes its color to red and immediately terminates. For $n=2$, robot $r_i$ changes its color to red when it becomes active for the first time and moves orthogonal to line $\overline{r_ir_j}$ that connects it to the only other robot $r_j$ it sees in $\cC(r_i)$. When $r_i$ later realizes that $|\cC(r_i)|$ is still 2 and $r_j.light=red$, it simply terminates. For $n=3$, when $r_i$ realizes that both of its neighbors in $\cC(r_i)$ have light set to red and are collinear with it, it moves orthogonal to that line and sets its light to red. The next time it becomes active, it finds itself at one of the corners and simply terminates as it sees all the other robots in the corners of the hull with light set to red.

For $n \geq 4$, let $a$ and $b$ be the two robots that occupy the corners of the line segment $\cH_0$ (i.e. the endpoint robots of $\cH_0$). Nothing happens until $a$ or $b$ is activated, setting its light to red, and moving orthogonal to $\cH_0$. After $a$ or $b$ moves, when another robot becomes active, it realizes that the configuration is not a line anymore and enters the normal execution of our algorithm. It is easy to see that after the line segment $\cH_0$ evolves into a polygonal shape, it never reverts to being a line. 

Finally, since our algorithm uses only two colors, the theorem follows. 
\end{proof}

It remains to analyze the number of rounds needed by our algorithm.

\begin{lemma}
\label{lemma:eligibleEdgesTotal}
After $O(n)$ rounds, the convex hull has grown enough in size to allow all $n$ robots to become corners. 
\end{lemma}
\begin{proof}
  Since in every round all corner robots move a distance of 1 along the bisector of their exterior angle, the length of the convex hull grows by at least 1 in every round. Note that when a robot becomes a corner, it moves outside the current convex hull and thus, by triangle inequality, extends the hull that way as well. 
  
  Hence, after at most $4n$ rounds the convex hull is long enough to ensure that there is space for all interior robots: there are at most $n$ edges of the convex hull and for each of them to \emph{not} be long enough, their total length is strictly less than $3n$. Hence, by expanding the convex hull by a total of $4n$, we ensure that there is enough space for each of the less than $n$ interior robots of diameter 1. Expanding the convex hull a total of $4n$ takes $O(n)$ rounds, completing the proof. 
\end{proof}

We note that for the above lemma the corner robots do not need to know $n$, as they can simply keep moving until the algorithm finishes. 

\begin{lemma}
\label{lemma:18}
The Interior Depletion phase of the mutual visibility algorithm finishes in $O(n)$ rounds.
\end{lemma}
\begin{proof}
When an interior robot can move outside the convex hull to become a corner robot, it needs at most a constant rounds to do so. During those rounds the robot becomes active, checks its path while moving to the safe zone to become a corner robot, and changes its light to red. There are fewer than $n$ interior robots and by Lemma~\ref{lemma:6} at least one robot can move when there is an edge of length at least 3. By Lemma~\ref{lemma:eligibleEdgesTotal} in $O(n)$ rounds there are sufficient long edges to allow the less than $n$ interior robots to move through them. Therefore, the Interior Depletion phase of the mutual visibility algorithm finishes in $O(n)$ rounds.
\end{proof}

We now have the following theorem bounding the running time of our algorithm using Lemmas~\ref{lemma:3} and~\ref{lemma:18} and Theorem~\ref{lemma:16}.

\begin{theorem}
\label{lemma:20}
Our algorithm solves the {\mv} problem for unit disk robots in $O(n)$ rounds without collisions in the fully synchronous setting using two colors. 
\end{theorem}

\section{Concluding Remarks}
We studied the {\mv} problem for a system of autonomous fat robots of unit disk size in the robots with lights model. We described an algorithm for this problem that uses two colors and works for fully synchronous computation of fat robots under rigid movements. Our solution is optimal with respect to the number of colors used since even for point robots at least two colors are required~\cite{Sharma2015}. Also, our algorithm solves the {\mv} problem in $O(n)$ rounds. For future work, it is interesting to extend our algorithm for non-rigid movements of robots and also for semi-synchronous and asynchronous computations.

\bibliographystyle{plain}
\bibliography{references}

\newpage
\appendix
\section{Pseudocode}
\begin{algorithm}[ht]
{
$//$ Look-Compute-Move cycle for robot $r_i$ of unit disk size\\
$\cC_k(r_i) \leftarrow$ configuration $\cC_k$ for robot $r_i$ (including $r_i$);\\
$\cH_k(r_i) \leftarrow$ convex hull of the positions of the robots in $\cC_k(r_i)$;\\
{\bf if} $|\cC_k(r_i)|=1$ {\bf then}
Terminate;\\
{\bf else if} $\cH_k(r_i)$ {\it is a line segment} {\bf then}\\
{\Indp
{\bf if}  $|\cC_k(r_i)|=2$ {\bf then}\\
{\Indp
 Let $r_j\in \cC_k(r_i)$;\\
{\bf if} $r_i.light =$ {\it Off} {\bf then} \\
{\Indp
Move orthogonal to line $\overleftrightarrow{r_ir_j}$ by any non-zero distance;\\
 $r_i.light\leftarrow$ {\it Red};\\
}
{\bf else if} $r_j.light=$ {\it Red} {\bf then}\\ 
{\Indp
{$r_i.light=$ {\it Red};\\
Terminate;} \\
}
}
{\bf else if} $|\cC_k(r_i)|=3$ {\bf then}\\
{\Indp
Let $r_j, r_l\in \cC_k(r_i)$;\\
{\bf if} $r_i.light=$ {\it Off} $\wedge$ $r_j.light=$ {\it Red} $\wedge$ $r_l.light=$ {\it Red} {\bf then} \\
{\Indp
Move orthogonal to line $\overleftrightarrow{r_jr_l}$ by any non-zero distance;\\
$r_i.light\leftarrow$ {\it Red};\\
}}}
{\bf else if} $r_i$ {\it is a corner robot of}  $\cH_k(r_i)$ {\bf then}
$Corner(r_i,\cC_k(r_i),\cH_k(r_i))$;\\
{\bf else if} $r_i$ {\it is an interior robot of}  $\cH_k(r_i)$ {\bf then}
$Interior(r_i,\cC_k(r_i),\cH_k(r_i))$;\\
{\bf else if} $r_i$ {\it is a side robot of}  $\cH_k(r_i)$ 
{\bf then}
$Side(r_i,\cC_k(r_i),\cH_k(r_i))$;\\
}
\caption{{\mv} algorithm 
} 
\label{algorithm:2}
\end{algorithm}

\begin{algorithm*}[ht]
{
\scriptsize
{\bf if} $r_i.light = $ {\it Off} {\bf then} \\
{\Indp
Order the robots in $\cH_k(r_i)$ starting from any arbitrary robot $v_1$ in the clockwise order so that $\cT=\{v_1,\ldots, v_{last},v_1\},$ where $v_1$ is the first robot and $v_{last}$ is the last robot;\\
Let $c,d$ be any pair of two consecutive robots in $\cT$ with
$c.light=$ {\it Red} and $d.light=$ {\it Red};\\
Let $HP_{cd}$ be the half-plane defined by line parallel to $\overleftrightarrow{cd}$ that passes through $r_i$ such that
$c,d$ are in $HP_{cd}$;\\
$Q \leftarrow$ set of line segments $\overline{cd}$ such that:\\
{\Indp
(a) the triangle $r_i,c,d$ does not contain (neither inside nor on its edges) any other robot of $\cC_k(r_i)$, and \\
(b) there is no robot in edge $\overline{cd}$, and \\
(c) there is no robot in $\cC_k(r_i) \backslash \cH_k(r_i)$ closer to edge $\overline{cd}$ than $r_i$, and\\
(d) there are no two robots with equal distance to $\overline{cd}$ appearing counterclockwise and clockwise of $r_i$ with respect to the local coordinate system of $r_i$, and\\
(e) the length of $\overline{cd}$ is at least 3\;
}
{\bf if} $Q$ {\it is not empty} {\bf then} \\
{\Indp
$\overline{u_1u_2} \leftarrow$ the line segment in $Q$ between two robots $u_1,u_2$ that is closest to $r_i$;\\
{\bf if} there is no other robot with light {\it Off} that is at equal distance to $\overline{u_1u_2}$  {\bf then}\\
{\Indp
$m\leftarrow$ midpoint of $\overline{u_1u_2}$\;
$L\leftarrow$ line perpendicular to $\overline{u_1u_2}$ passing through its midpoint $m$\;
Order the robots in the counterclockwise order of $r_i$ (with respect to the local coordinate system of $r_i$) such that the order is $\cT_i=\{v_1,v_2,v_3,v_4\}$;\\
Compute angles $\alpha = 180^\circ-\angle v_4v_3v_2$ and $\beta = 180^\circ-\angle v_1v_2v_3$, and set $\delta =\min\{\alpha/4,\beta/4\}$;\\
Compute a point $x'$ such that $\angle x'v_3v_2 = \delta$ and a point $x''$ such that $\angle x''v_2v_3 = \delta$;\\
$L(r_im) \leftarrow$ line segment connecting $r_i$ and $m$\;
$x=x'+L(r_im)$, where $x'$ is the nearest point to $e$ in the safe zone outside the convex hull\;
$Move(r_i,\cC_k(r_i),\cH_k(r_i),u_1,u_2,x)$\;
}
{\bf else if} there exists a robot in the clockwise direction of $r_i$ (with respect to the local coordinate system of $r_i$) with light {\it Off} that is at equal distance to $\overline{u_1u_2}$ {\bf then}\\
{\Indp
$m\leftarrow$ point in $\overline{u_1u_2}$ at  $\frac{\length(\overline{u_1u_2})}{3}$ from endpoint $u_1$\;
$L\leftarrow$ line perpendicular to $\overline{u_1u_2}$ passing through the point $m$\;
Order the robots in the counterclockwise order of $r_i$ (with respect to the local coordinate system of $r_i$) such that the order is $\cT_i=\{v_1,v_2,v_3,v_4\}$;\\
Compute angles $\alpha = 180^\circ-\angle v_4v_3v_2$ and $\beta = 180^\circ-\angle v_1v_2v_3$, and set $\delta =\min\{\alpha/4,\beta/4\}$;\\
$L(r_im) \leftarrow$ line segment connecting $r_i$ and $m$\;
Compute a point $x'$ such that $\angle x'v_3v_2 = \delta$ and a point $x''$ such that $\angle x''v_2v_3 = \delta$;\\
$x=x'+L(r_im)$, where $x'$ is the nearest point to $e$ in the safe zone outside the convex hull\;
$Move(r_i,\cC_k(r_i),\cH_k(r_i),u_1,u_2,x)$\;
}
{\bf else if} there exists a robot in the counterclockwise direction of $r_i$ (with respect to the local coordinate system of $r_i$) with light {\it Off} that is at equal distance to $\overline{u_1u_2}$ {\bf then}\\
{\Indp
$m\leftarrow$ point in $\overline{u_1u_2}$ at  $\frac{\length(\overline{u_1u_2})}{3}$ from endpoint $u_2$\;
$L\leftarrow$ line perpendicular to $\overline{u_1u_2}$ passing through the point $m$\;
Order the robots in the counterclockwise order of $r_i$ (with respect to the local coordinate system of $r_i$) such that the order is $\cT_i=\{v_1,v_2,v_3,v_4\}$;\\
$L(r_im) \leftarrow$ line segment connecting $r_i$ and $m$\;
Compute angles $\alpha = 180^\circ-\angle v_4v_3v_2$ and $\beta = 180^\circ-\angle v_1v_2v_3$, and set $\delta =\min\{\alpha/4,\beta/4\}$;\\
Compute a point $x'$ such that $\angle x'v_3v_2 = \delta$ and a point $x''$ such that $\angle x''v_2v_3 = \delta$;\\
$x=x'+L(r_im)$, where $x'$ is the nearest point to $e$ in the safe zone outside the convex hull\;
$Move(r_i,\cC_k(r_i),\cH_k(r_i),u_1,u_2,x)$\;
}
}
}
}
\caption{$Interior(r_i,\cC_k(r_i),\cH_k(r_i))$}
\label{algorithm:2i}
\end{algorithm*}

\begin{algorithm}[ht]
{
Move $r_i$ distance 1 along the angle bisector of its neighbors on $\cC_k(r_i)$ in the direction that does not intersect the interior of $\cC_k(r_i)$;\\
{\bf if} $r_i.light=$ {\it Off}
{\bf then}
$r_i.light \leftarrow$ {\it Red};\\
{\bf else if} $\forall r\in \cC_k(r_i)$, $r.light=$ {\it  Red}
{\bf then}  Terminate;\\

\caption{$Corner(r_i,\cC_k(r_i),\cH_k(r_i))$}
\label{algorithm:2v}
}
\end{algorithm}

\begin{algorithm}[ht]
{
$L_{r_ix} \leftarrow$ line segment connecting $r_i$ and $x$\;
$L'_{r_ix}, L''_{r_ix} \leftarrow$ lines parallel to $L_{r_ix}$ at distance $1/2$ on either side of $L_{r_ix}$ towards $\overline{u_1u_2}$\;
{\bf if} $L'_{r_ix}$ and $L''_{r_ix}$ share no point occupied by any other robot {\bf then}\\
{\Indp
Move to point $x$ in the safe zone\;
$r_i.light \leftarrow $ {\it Red}\;
}}
\caption{$Move(r_i,\cC_k(r_i),\cH_k(r_i),u_1,u_2,x)$}
\label{algorithm:move}
\end{algorithm}

\begin{algorithm}[ht]
{
{\bf if} {\it at least one neighbor of $r_i$ in the edge $e$ it belongs to has light} {\it Red} {\bf then}\\
{\Indp
Order the robots in the counterclockwise order of $r_i$ (with respect to the local coordinate system of $r_i$) such that the order is $\cT_i=\{v_3,v_2,r_i,r,v_0\}$, where $v_3$ is the first robot non-collinear to $r_i$ in the clockwise direction of $r_i$ with $v_3.light=$ {\it Red}, $v_2$ is the robot that is collinear with $r_i$ in the clockwise direction of $r_i$, and $r$ is the collinear robot in the counterclockwise direction of $r_i$, and $v_0$ is the first non-collinear robot to $r_i$ in the counterclockwise direction of $r_i$ with $v_0.light=$ {\it Red};\\
Compute angles $\alpha = 180^\circ-\angle v_0rr_i$ and $\beta = 180^\circ-\angle r_iv_2v_3$, and set $\delta =\min\{\alpha/4,\beta/4\}$;\\
Compute a points $x'$ and $x''$ such that $\angle x'v_2r_i = \delta$ and $\angle x''rr_i = \delta$ and $r_ix'$ and $r_ix''$ are perpendicular to $e$;\\
$x\leftarrow$ $x'$ or $x''$ whichever is nearest to $e$;\\
Move perpendicular to $e$ with destination $x$;\\
$r_i.light \leftarrow $ {\it Red}\;
}}
\caption{$Side(r_i,\cC_k(r_i),\cH_k(r_i))$}
\label{algorithm:2s}
\end{algorithm}

\end{document}